\documentclass{article}

\usepackage[english]{babel}

\usepackage[letterpaper,top=2cm,bottom=2cm,left=3cm,right=3cm,marginparwidth=1.75cm]{geometry}

\usepackage{amsmath}
\usepackage{graphicx}
\usepackage[colorlinks=true, allcolors=blue]{hyperref}
\usepackage{authblk}
\usepackage[english]{babel}
\newcommand{\keywords}[1]{\par\textbf{Keywords:} #1\par}
\usepackage{amsthm}
\newtheorem{lemma}{Lemma}[section] 
\newtheorem{theorem}{Theorem}[section]

\title{The Generalized Word Count in Two-Level Fractional Factorial Designs}

\author{Xietao Zhou}
\affil{\small Department of Mathematics, King's College London}
\affil{\small Strand, London WC2R 2LS, UK.}
\affil{\small xietao.zhou@kcl.ac.uk}

\author{Steven G.\ Gilmour}
\affil{\small Department of Mathematics, King's College London}
\affil{\small Strand, London WC2R 2LS, UK.}
\affil{\small steven.gilmour@kcl.ac.uk}

\date{} 

\begin{document}
\maketitle

\begin{abstract}
Unreplicated two-level factorial designs are often used in screening experiments to determine which factors out of a large plausible set are active. A theorem regarding the generalized word count pattern is stated and proved for unreplicated designs. It is shown that a phenomenon regarding optimal designs seen in the recent literature can be explained by the theorem obtained.
\end{abstract}
\section{Introduction}

Factorial designs of various types are useful when experimenters are studying the effects of several factors on one or more responses and they have seen wide application in industrial research and many other areas. Many criteria have been suggested to evaluate designs and among them, the minimum aberration criterion was proposed in \cite{7c8d897c-7cec-366b-9bba-57448ea51aaa} to minimize the aliasing between factorial effects. This criterion only works for regular designs and for non-regular designs, minimum $G_2$ aberration was introduced in \cite{deng1999minimum} which criterion sequentially minimizes the sequence of metrics, $B_s(D)$, $s=1, \ldots,m$ for $m$ factors, which are the corresponding summations of squares of normalized J-characteristics defined in that paper. The metrics, $B_s(D)$, were generalized to the concept of generalized word count $b_s(i_1,\cdots,i_p)$ in \cite{Gilmour2010} to measure aliasing of all effects with $s$ factors having the same number of factors at a particular order when polynomial effects are considered. When restricted to two-level designs, $B_s(D)$ is equivalent to $b_s(s)$, a measure of aliasing between the intercept and interactions with $s$ factors in \cite{Gilmour2010}. For the rest of the paper, we will simply use $b_s$ as the notation of the generalized word counts. We will give the definition of the generalized word count $b_s$ in section~\ref{sec:OverviewGwc} and as a part of this definition, a definition equivalent to that of the J-characteristics in \cite{deng1999minimum} will be given.

Various aspects of minimum $G_2$ aberration have been explored in the existing literature. The properties of J-characteristics and the projection justification of minimum $G_2$ aberration were discussed in \cite{27196536-5cf4-3f7b-bc60-2f14e0732c03}. The projection aspects of minimum $G_2$ aberration were further explored in \cite{4e0135ef-cb45-3990-b0f6-85b091a22833} by associating it with the estimation capacity proposed by \cite{10.1111/1467-9868.00164}. The set of $b_s$ were associated with the $Q_B$ criterion in \cite{Gilmour2010} to use prior probabilities of main effects and two factor interactions being in the final fitted model to better reflect experimenters' beliefs compared to other criteria. The process of searching for $Q_B$ optimal designs was discussed in \cite{TG2024}. Lower bounds on the maximum of J-characteristics have been proved and used to construct designs that maximise the generalized resolution proposed in \cite{0cd53606-262f-394a-93e3-9e0a0d082560} from orthogonal arrays in \cite{07f1dda9-a7d9-3b4a-8288-ecfa8188fb28}. The  minimum $G_2$ aberration itself has also been used in \cite{Hameed08042025a} as a secondary criterion to further distinguish the $D$- and $A$-optimal designs they found for run sizes one or two more than a multiple of four.

In this paper we aim to understand the generalized word count $b_s$ from a new angle, rather than assuming unrelated and scattered values of each $b_s$ for $s \in \{1,\ldots,m\}$. We show that for unreplicated two-level factorial designs the sum over all $b_s$ is a constant which is jointly determined by the number of factors and experimental treatments. The theorem could be of interest for practical application since unreplicated fractional factorial designs are effective in isolating influential factors as in \cite{cd3c1127-41ea-360d-a880-30452b921020}. A specific phenomenon found in \cite{Hameed08042025b} can be explained by our theorem. Throughout the rest of this paper, we use the usual centered parameterization, with levels coded as $-1$ and $1$.



The rest of this paper is organized as follows. An overview of generalized word counts is given in Section~\ref{sec:OverviewGwc}, followed by the statement and proof of our main theorem in Section~\ref{sec:ThmProof}. The discussion of the phenomenon noted in \cite{Hameed08042025b} and other possible applications of this theorem can be found in Section~\ref{sec:ApplicationThm}.

\section{Overview of the generalized word counts}
\label{sec:OverviewGwc}
The generalized word count, $b_s(i_1,\cdots,i_p)$, was first introduced in \cite{Gilmour2010} as a generalization of the metric $b_s$ in \cite{deng1999minimum} to measure aliasing of all effects with $s$ factors having the same number of factors at a particular order when polynomial effects are considered. As an example, in a second order linear model for which linear main effects, quadratic main effects and two-factor interactions are considered, $b_3(1,2)$ measures the aliasing between the intercept parameter and all the interactions involving three factors, the linear effect of one and the quadratic effects of the other two, while $b_3(3,0)$ is equivalent to $b_3$ as if there were no quadratic effects.

When quadratic effects are not considered the definition of $b_s$ can be found in \cite{TG2024} as 
\begin{equation}
\label{3.101}
b_s=\sum_{|g|=s} R_s(g) ; \quad R_s(g)=\frac{1}{N^2}\left[\sum_{h=1}^N\left(X_{h 1} \cdots X_{h g}\right)\right]^2,
\end{equation}
where $X_1, \cdots, X_m$ are all the columns forming the design, $g$ is a subset of $s$ columns of this full set and the J-characteristics of \cite{deng1999minimum} can be given as
\begin{equation}
\label{3.102}
 J_s(g)=\sum_{h=1}^N\left(X_{h 1} \cdots X_{h g}\right).
\end{equation}
As mentioned in the introduction, the set of $b_s$ was used one by one sequentially in \cite{deng1999minimum}. Later, components of this set were combined as linear combinations to define a new criterion motivated by estimation capacity by \cite{4e0135ef-cb45-3990-b0f6-85b091a22833}. This criterion, however, is limited to scenarios in which all main effects are expected to be active, which might not be true in a practical situation. The $Q_B$ criterion defined in \cite{Gilmour2010} deals with this potential issue by allowing flexible prior probabilities of pairs of main effects and interactions to fully enable experimenters to use their prior beliefs about the experiments. With the aforementioned definitions of $b_s$ and $J_s(g)$ we prove our main theorem in the next section.

\section{Main theorem}

\label{sec:ThmProof}
The generalized word count $b_{s}$ can be 
shown to have the following properties.

\begin{theorem}\label{mainThm}
For an $N$ run design for $m$ factors, if all $N$ runs have distinct treatments, then $$\sum_{s=1}^m b_{s}=C_{m,N}, $$
where $C_{m,N}$ is a constant which depends only on the number of factors and number of runs and not on the design.
\end{theorem}

We will start by proving a sequence of lemmas before Theorem~\ref{mainThm} can be proved. The first result is about the summation of elements in the rows of a full factorial design matrix, i.e.\ the matrix contains all possible main effects and interaction columns, and can be stated as follows.

\begin{lemma}\label{lem:1}
In a full factorial design matrix with $m$ factors, considering the sub-matrix excluding the constant column, the row sum of elements will always be $-1$ for each row except for the row $[1, \cdots, 1]$, for which the sum of elements would be $2^m-1$. 
\end{lemma}

\begin{proof}
Let $n=2^m$. We shall proceed by induction, It is trivial to see that the lemma is valid in the two cases with $m=1$ and $m=2$. We then assume that the lemma is true for a full factorial design matrix of $m$ factors. The structure of such a matrix is $\left[\begin{array}{lll}\textbf{1} & D_m & S_m\end{array}\right]$, where $ D_m$ and $S_m$ represent the sub matrices of main effects and interactions respectively. The sum of each row will then be given by $\left[\begin{array}{c}-1 \\ \vdots \\ -1 \\ n-1\end{array}\right]$.

We now consider adding a new factor as the induction step. The structure of the design matrix will now become
\begin{equation}\label{equation1}
\left[\begin{array}{cccccc}
\textbf{1} & \textbf{-1} & D_m & S_m & -D_m & -S_m \\
\textbf{1} & \textbf{1} & D_m & S_m & D_m & S_m
\end{array}\right].
\end{equation}
Taking advantage of the assumption, the sum of elements from each row can be given as
\begin{equation}
\left[\begin{array}{c}
\textbf{0} \\
\textbf{2}
\end{array}\right]+\left[\begin{array}{c}
-1 \\
\vdots \\
n-1 \\
-1\\
\vdots \\
n-1
\end{array}\right]+\left[\begin{array}{c}
1 \\
\vdots \\
-(n-1) \\
-1 \\
\vdots \\
n-1
\end{array}\right]=\left[\begin{array}{c}
0 \\
\vdots \\
0 \\
2 n
\end{array}\right]=\textbf{1}+\left[\begin{array}{c}
-1 \\
\vdots \\
-1 \\
2 n-1
\end{array}\right].
\end{equation}
Again, if we set aside the constant column, $$\left[\begin{array}{c}-1 \\ \vdots \\ -1 \\ 2n-1\end{array}\right]$$ gives exactly the desired result for the updated run size $2n$ with $m+1$ factors and thus the lemma is proved.
\end{proof}

Staying with the setup of a full factorial design with $m+1$ factors, we now prove another lemma.

\begin{lemma}
    
\label{lem:2}
 Assuming $m+1$ factors, consider the row $[-1,1, \cdots, 1]$ where the $-1$ corresponds to the row element from factor column $X_i$. Then in rows other than $[-1,1, \cdots, 1]$ and $[1, \cdots, 1]$, the sum of row elements of all effects associated with factor $X_i$ will be 0. For row $[-1,1, \cdots, 1]$ and $[1, \cdots, 1]$, the result is given by $-n$ and $n$ respectively.
\end{lemma}

\begin{proof}
We consider again the structure of the design matrix of $m+1$ factors. The sub-matrix that is due to the introduction of the new factor consists of the second column and the last two columns of equation~(\ref{equation1}). Applying Lemma~\ref{lem:1} we have
$$
\left[\begin{array}{cc}
-1 & 1 \\
\vdots & \vdots \\
-1 & -(n-1) \\
1 & -1 \\
\vdots & \vdots \\
1 & n-1
\end{array}\right]=\left[\begin{array}{c}
0 \\
\vdots \\
0 \\
-n \\
0 \\
\vdots \\
0 \\
0 \\
n
\end{array}\right].
$$
Note that $[-1, -(n-1)]$ corresponds to $[-1, 1,\cdots, 1]$ and $[1, (n-1)]$ corresponds to $[1, \cdots, 1]$ and the rest of the rows give the desired row elements summing to 0 and the lemma is then proved.
\end{proof}

We are now ready to state and prove a special case of  Theorem~\ref{mainThm} as follows.

\begin{lemma}\label{lem:3}
Consider a design for $m$ factors with $N$ distinct treatments, including the row $[-1,1, \cdots, 1]$ and excluding $[1, \cdots, 1]$. The first coordinate from the treatment $[-1, 1,\cdots, 1]$ is being considered for exchange. Then before and after the exchange the sum of the generalized word counts, $b_{s}$, will be a constant.
\end{lemma}

\begin{proof}
For the coordinate exchange described in the statement of the lemma, without loss of generality we assume the position comes from factor column $X_1$. The $J_s(g)$ value will increase by $2$ for all effects which include $X_1$ and remain unchanged for other effects and there are $2^{m-1}$ such effects which include $X_1$.

Each of these terms will contribute to the increase of the sum of the generalized word counts, $b_{s}$ by $4J_s(g)+4$ up to the scaling factor $\frac{1}{N^2}$. The lemma then suggests the summation of all these contributions is equal to 0. This means we have to verify the summation of all $J_s(g)$ associated with $X_1$ are equal to $-2^{m-1}$.

According to the definition of the J-characteristic, the summation of $J_s(g)$ of all effects including $X_1$ is the summation of $J_s(g)$ of each corresponding effect column, calculated by summing the matrix elements in that column. For the  set of all effect columns including $X_1$, the summation described can also be done equivalently row by row, for which we could apply the result of Lemma~\ref{lem:2}. 

Follow the setup of this lemma, we do the summation of matrix elements row by row and exclude row $[1, \cdots, 1]$ before the coordinate exchange. For $m+1$ factors that sum equals $-n=-2^{m}$ by Lemma~\ref{lem:2}, and thus for $m$ factors that sum equals $-2^{m-1}$, which completes the proof.  
\end{proof}

Now that we have got a special case of Theorem~\ref{mainThm}, we will proceed to the full proof by arguing that a general coordinate exchange can be transformed into our special case by row switches, column switches and multiplying whole columns by $-1$, all of which will not change the sum of the generalized word counts, $b_{s}$. The switch operations are trivial, thus we will just give a lemma regarding the operation of multiplying a whole column by $-1$ as follows.

\begin{lemma}

\label{lem:4}
For a valid design which satisfies the setup in Theorem~\ref{mainThm}, the results still hold true after an arbitrary column is multiplied by $-1$.
\end{lemma}

\begin{proof}
The effect of such an operation is to negate all those $J_s(g)$ associated with that particular factor to $-J_s(g)$ and do nothing to those unrelated $J_s(g)$. But for each changed $J_s(g)$, $\frac{J_s(g)^2}{N^2}$ and $\frac{(-J_s(g))^2}{N^2}$
contribute equally to the sum of the generalized word counts, $b_{s}$. Therefore the corollary is proved.
\end{proof}

The proof of Lemma~\ref{lem:4} now means that the whole sequence of lemmas centered around Theorem~\ref{mainThm} have been proved, and that leads to the proof of Theorem~\ref{mainThm} itself.

\section{Discussion}

\label{sec:ApplicationThm}
\cite{Hameed08042025b} used the minimum G-aberration criterion proposed in \cite{0cd53606-262f-394a-93e3-9e0a0d082560} as a secondary criterion to further distinguish the D-optimal designs they found for experiments with run size 3 more than a multiple of 4. The minimum G-aberration criterion uses the same set of information regarding $J_s(g)$ and $R_s(g)$ as the minimum $G_2$ aberration. Instead of calculating $b_s$, however, the minimum G-aberration criterion sorts the frequency of the J-characteristics of $s$-tuples and arranges it into a confounding frequency vector $F$ with components $F_1,F_2,\cdots$ which serve a similar purpose as $b_1,b_2,\cdots$. Therefore our arguments work similarly for the minimum G-aberration criterion as well, for which we will illustrate with a worked example.

\subsection{Worked example}

\label{sec:workingEg}
\cite{Hameed08042025b} commented that for 9 factors and 15 runs, the design which has the best $F_3$ vector has the worst $F_4$ vector. Since the minimum G-aberration looks at the same type of aliasing between effects as the minimum $G_2$ aberration, through our Theorem~\ref{mainThm} it is clear that this phenomenon can be explained by our theorem. Similar results appeared for the designs they found for 9 factors and 19 treatments. The minimal G-aberration D-optimal design they found, has achieved a significant reduction in aliasing between the main effects and two-factor interactions measured by the $F_3$ vector compared to the default JMP design. In return, however, it has introduced a lot of aliasing between two-factor interactions shown by the components of their $F_4$ vector. In fact, the generalized word count $b_4$ of the JMP design is much better than the design they obtained. This  suggests two things: 
\begin{enumerate}
    \item Our theorem successfully explains the reasons behind the observed phenomenon and provides a more general theoretical explanation. 
    \item It shows that carefully combining the information across all available $b_s$ to choose the design seems a reasonable strategy in terms of being closer to the diverse needs arising from actual experiments. To that end, the $Q_B$ criterion in \cite{Gilmour2010} is a great choice.
    \end{enumerate}

\subsection{Open Questions}
\label{sec:OpenQues}
The ways Theorem~\ref{mainThm} can be further applied remains an open question. One possibility might be try to maximize $b_5$ and higher by construction. Governed by Theorem~\ref{mainThm}, the sum of $b_1$ to $b_4$ will then go down to the minimal possible value. Combined with state-of-art constrained optimization techniques, it might then be possible to achieve some lower bound values for criteria like $Q_B$ and the projection type of criterion described in \cite{4e0135ef-cb45-3990-b0f6-85b091a22833}. This might be especially interesting for large numbers of runs and factors in an experiment when the mixed integer quadratic programming (MIQP) algorithm, proposed in \cite{article}, which verifies that a given design is optimal, ceases to converge. Theorem~\ref{mainThm} might then be applied as an important building block to check if the algorithm planned for use can do a reasonable job in searching for optimal designs. The details of these possibilities are left for future research.

\section*{Acknowledgement}
The first author expresses sincere thanks for funding through the King's-China Scholarship Council Scholarship program (File No. 202108060084).

\bibliographystyle{apalike}
\bibliography{reference}
\end{document}